\documentclass[letterpaper, 10 pt, conference]{ieeeconf}  % Comment this line out
\IEEEoverridecommandlockouts                              % This command is only
\usepackage{cite} 
\usepackage{graphicx}
\usepackage{textcomp}
\usepackage{graphicx} 
\usepackage{float}% include this line if your document contains figures
\usepackage[nopar]{lipsum}

\usepackage{amsmath,amssymb,amsfonts,xpatch,dsfont}
\newcommand{\R}{\mathbb{R}}

\newcommand{\E}{\mathbb{E}}

\newcommand{\argmin}{\mathop{\rm arg~min}\limits}
\DeclareMathOperator{\suth}{s.\;t.}

\newcommand{\diag}{\text{diag}}
\DeclareMathOperator{\kernel}{ker}
\DeclareMathOperator{\image}{im}
\DeclareMathOperator{\one}{\mathds{1}}
\DeclareMathOperator{\vectorize}{vec}
\newcommand{\unob}{\eta_{\bar{o}}}
\newcommand{\obs}{\eta_{o}}
\newcommand{\unobpre}{\hat{\eta}^{-}_{\bar{o}}}
\newcommand{\obspre}{\hat{\eta}^{-}_{o}}
\newcommand{\unobest}{\hat{\eta}_{\bar{o}}}
\newcommand{\obsest}{\hat{\eta}_{o}}
\newcommand{\errorobspre}{P^{-}_{oo}}
\newcommand{\errorunobpre}{P^{-}_{\bar{o}o}}
\newcommand{\errorobsest}{P_{oo}}
\newcommand{\errorunobest}{P_{\bar{o}o}}
\newcommand{\weight}{q}
\newcommand{\obsgain}{L_{o}}
\newcommand{\unobgain}{L_{\bar{o}}}
\newcommand{\mat}[1]{\left[\; \begin{matrix} #1 \end{matrix} \:\right]}

\newcommand{\Hada}{\mathsf{H}}
\newcommand{\simode}[1]{\left\{\;
	\begin{aligned} #1 \end{aligned} \right.}
\newcommand{\abs}[1]{\left\lvert{#1}\right\rvert} %absolute value

\newtheorem{theorem}{Theorem}

\newtheorem{remark}{Remark}
\newtheorem{lemma}{Lemma}

\title{\LARGE \bf
	Explicit Ensemble Mean Synchronization for Time Scale Generation with Mixed Atomic Clock Ensembles*
}

\author{Priyanka Dey, %\IEEEmembership{Member, IEEE}, 
	Takahiro Kawaguchi, %\IEEEmembership{Member, IEEE}, 
	Yuichiro Yano, %\\ 
	Yuko Hanado,   and Takayuki Ishizaki %,  \IEEEmembership{Member, IEEE% <-this % stops a space
		\thanks{*This work is supported by the Ministry of Internal Affairs and Communications (MIC) under its "Research and Development for Expansion of Radio Resources (JPJ000254)" program.}% <-this % stops a space
		\thanks{P. Dey and T. Ishizaki are with the Department of Systems and Control Engineering, Institute of Science Tokyo, Meguro, Tokyo 152-8552 Japan (Emails: {\tt dey.p.aa@m.titech.ac.jp, ishizaki@sc.eng.isct.ac.jp}). T. Kawaguchi is with the Division of Electronics and Informatics, Gunma University, Kiryu, Gunma 371-8510 Japan (Email: {\tt kawaguchi@gunma-u.ac.jp}). Y. Yano and Y. Hanado are with the National Institute of Information and Communications Technology, Koganei, Tokyo 184-0015 Japan (Emails: {\tt y-yano@nict.go.jp, yuko@nict.go.jp}).}%
	}
	
\begin{document}

		\maketitle
		\thispagestyle{empty}
		\pagestyle{empty}

		%%%%%%%%%%%%%%%%%%%%%%%%%%%%%%%%%%%%%%%%%%%%%%%%%%%%%%%%%%%%%%%%%%%%%%%%%%%%%%%%
		\begin{abstract}
			In this paper, we consider a mixed ensemble containing a mixture of cesium-type and hydrogen maser-type atomic clocks. 
			For the mixed ensemble, the conventional Kalman filtering algorithm has certain limitations due to divergence of the error covariance matrix. 
			To overcome these limitations, we obtain a Kalman filtering algorithm based on observable canonical decomposition that does not have any diverging terms. 
			We use the estimates from the transformed Kalman filter to propose a time scale generation algorithm called explicit ensemble mean synchronization algorithm for the mixed ensemble. 
			In this algorithm, we synchronize the time deviation of each clock from the ideal clock behavior to the unobservable ensemble mean of the phases where the weighting can be decided by the user.
			By regulating the free-running dynamics associated with the unobservable state, through choosing an appropriate weight vector, the frequency stability of the generated time scale or the synchronized time shared by the clocks is optimized over shorter (resp. longer) intervals, as measured by Hadamard variance. 
			An illustrative example is given to demonstrate the efficiency of our algorithm.
			
		\end{abstract}

		%%%%%%%%%%%%%%%%%%%%%%%%%%%%%%%%%%%%%%%%%%%%%%%%%%%%%%%%%%%%%%%%%%%%%%%%%%%%%%%%
		\section{Introduction}
		\label{s:Introduction}
		\noindent An atomic clock ensemble is a group of highly precise atomic clocks of different types, such as cesium and hydrogen maser clocks, that operate collectively to provide a reliable and accurate timing infrastructure \cite{galleani2010time}.
		Atomic clocks are devices that measure time based on the vibrations of atoms with constant resonance frequencies.
		The national metrology institute of each country uses the clock difference measurements from several atomic clocks within an ensemble to reduce these vibrations of atoms to create a reliable and robust timekeeping system \cite{ref:HanImaKot-07}.
		This development of accurate time scales provides significant advantages across several critical sectors, including satellite
		navigation, financial networks with high-frequency trading and time-sensitive transactions, smart grids, etc \cite{Bandi2023compresive}.
		The fluctuations in the tick rates of atomic clocks are referred to as time deviations from ideal clock behavior, which can be modeled as a stochastic process. 
		Based on this stochastic nature of time deviation, the cesium-type (Cs-type) clock has a second-order model \cite{galleani2008atutorial} while the hydrogen maser-type (Hm-type) clock has a third-order model with an additional frequency drift state \cite{zucca2005theclock}.
		
		In the task of time scale generation, the time deviation of each clock in the ensemble is estimated to correct the reading of the clock either by software offset or by physically steering the clock. 
		Conventional Kalman filter (CKF) is extensively used to accomplish this estimation \cite{Brown1991theory, galleani2003use}.
		However, the system model of atomic clock ensemble is undetectable because of its inherent characteristic that only clock reading differences can be measured.
		As a result, the error covariance matrix grows boundlessly, which may lead to numerical instability.
		However, the Kalman gain converges to a constant value. 
		This indicates that the CKF has some indeterminate operations in the algorithm.
		Some works proposed methods to add a modification step to the error covariance matrix in the CKF algorithm \cite{greenhall2006kalman, Brown1991theory, suess2013congruence}.
		The time scale is obtained by \emph{implicitly} taking the weighted average of the corrected clock readings termed as implicit ensemble mean (IEM) \cite{greenhall2006kalman, greenhall2012review}.%, Yan2014Structured
		
		IEM has better accuracy and precision than individual clock in the ensemble. 
		In time and frequency community, frequency stability quantifies accuracy and precision of the time scale as it describes how well the frequency is maintained over a certain period of time \cite{lombardi2002fundamentals}.
		Hadamard variance (HVAR) is used to measure the frequency stability of clocks and time scales \cite{ishizaki2023higher-order, Riley2008handbook}.
		A smaller value of HVAR over an interval indicates that the time scale has reduced divergence rate from the ideal time over that interval.
		%In this paper, we use stability to mean frequency stability unless otherwise specified.
		Cs-type clock has excellent frequency stability for long intervals, while exhibiting a lower frequency stability for short intervals while Hm-type clock has excellent short-term frequency stability but lower long-term frequency stability \cite{galleani2010time, ishizaki2023higher-order}.
		
		Our group has proposed a time scale generation algorithm from an ensemble containing only Cs-type clocks \cite{Ishizaki2025Explicit}. 
		The algorithm includes decomposing the ensemble into an observable component, representing the synchronization error, and an unobservable component which symbolizes the weighted average clock states termed as the generated time scale. 
		%The observable part  and the unobservable part 
		However, due to existence of frequency drift in Hm-type clocks, the technique involved in this algorithm \cite{Ishizaki2025Explicit} needs careful modifications and cannot be applied as it is for a mixed ensemble. 
		In \cite{Dey2024Clock}, we developed a new steering technique for arbitrary-order atomic clocks based on output stabilization for tracking of a virtual reference time such as GPS composite clock \cite{Brown1991theory} by a single follower atomic clock. %, paper clock \cite{trainotti2019simulating}
		Recently, we proposed an algorithm to generate a time scale with excellent long-term frequency stability by using only the ensemble of Hm-type clocks \cite{Dey2025Steering}. 
		A separate ensemble of Cs-type clocks is used solely to optimize the long-term frequency stability \cite{Dey2025Steering}. 
		We developed a time synchronization algorithm for mixed ensemble in \cite{Dey2024Synchronization}. 
		However, neither we have not paid attention on how this work can be exploited to introduce a technique for generating a time scale nor we concentrated on the frequency stability aspect, which is extremely important for time scale generation.

		In this paper, we introduce a novel time scale generation algorithm, which we call explicit ensemble mean synchronize algorithm for a mixed ensemble containing Cs-type and Hm-type clocks.
		Firstly, we obtain a Kalman filtering algorithm based on the observable canonical decomposition of the mixed ensemble which does not contain any diverging terms.
		Secondly, using the observable state estimates from the transformed Kalman filter, we regulate the frequency of the clocks to synchronize the time deviations of clocks to the \emph{unobservable} weighted average phase state, where the weighting can be \emph{explicitly} specified by the user. 
		In order to optimize the short-term/long-term frequency stability of the synchronized time or the generated time scale (GTS), we consider the free-running dynamics associated with the unobservable state. 
		We evaluate the HVAR of the free-running dynamics to find appropriate weight vector to improve the short-term (resp. long-term) frequency stability of the generated time scale.
		%Furthermore, it has been revealed \emph{numerically} that by using the unobservable state estimates from steady state transformed Kalman filter, we can design a periodic control to simultaneously optimize the long-term and short-term frequency stability of the resulting time scale.
		% In \cite{Dey2024Synchronization}, we observed that the mixed atomic clock ensemble is undetectable. 
		% As shown in, applying standard Kalman filter for an undetectable system leads to growing of the error covariance unboundedly, which might lead to numerical instability.
		% However, the standard Kalman filter gain converges to a constant values.
		% It means that the standard Kalman filter has some indeterminate operations in the algorithm
		
		The rest of the paper is organized as follows. 
		Section \ref{s:atomicclockmodel} describes the model of the mixed ensemble and presents the problems considered in the paper.
		In Section \ref{s:preliminaryresults}, we obtain the transformed Kalman filter and its steady-state version, which is free from any diverging terms, by employing observable canonical decomposition on the mixed ensemble. 
		We present our explicit ensemble mean synchronization algorithm for generating a time scale in Section \ref{s:solutionofProb2}.
		We illustrate the efficacy of our proposed algorithm through a numerical example in Section \ref{s:example}.% including a discussion section to observe numerically that unobservable state estimates can be used to optimize frequency stability of the generated time scale.
		
		\noindent \textbf{Notations}: We denote the set of real numbers and positive real numbers by \(\R\) and \(\R^{+}\), respectively. 
		%For \(r\in \N\), we let \([r]=\lbrace 1,2,3,\ldots,r\rbrace\). 
		For a matrix \(A\in \R^{n\times m}\), its transpose is \(A^{\top}\), while its kernel and image are \(\kernel A\) and \(\image A\), respectively. %, \(\rank(A)\) is its rank.
		We denote Kronecker product by \(\otimes\), the \(n\)-dimensional all-ones column vector by  \(\one_{n}\), and the identity matrix of dimension \(n\) by \(I_n\). 
		We denote the zero matrix of dimension \(n\times m\) by \(0_{n\times m}\) or \(0\) when the dimension is clear. 
		A diagonal matrix with diagonal entries as \(\lbrace c_1, c_2,\ldots, c_n\rbrace\) is denoted by \(\text{diag}(c_1,c_2, \ldots, c_n)\). The expected value of a random element is denoted by \(\E[\cdot]\).
		
		%%%%%%%%%%%%%%%%%%%%%%%%%%%%%%%%%%%%%%%%%%%%%%%%%%%%%%%%%%%%%%%%%%%%%%%%%%%%%%%%%%%%%%%%%%%%%%%%%%%%%%%%%%%%%%%%%%%%%%%%%%%%%%%%%%%%%
		\section{Mixed Ensemble Model and Problem Statement}
		\label{s:atomicclockmodel}
		\noindent We consider a mixed ensemble of \(N\) clocks that contains a mixture of Cs-type and Hm-type clocks, where the last \(M > 0\) clocks are Hm-type clocks.
		Consider a discrete-time sequence \(\lbrace t_0, t_1, \ldots \rbrace\) in an ideal time scale where we assume that the time difference between two adjacent times is a constant sampling interval \(\tau\) i.e., \(t_k = \tau k\), for \(k= 0,1, \ldots\)
		For brevity, we will use
		\[
		A=\mat{1 & \tau \\ 0 & 1},\; \beta=\mat{\frac{\tau^2}{2}\\ \tau},\; B=\mat{\tau \\ 1},\; C=\mat{1 & 0}.
		\]
		The discrete-time clock deviation model of a Cs-type atomic clock \(i \in \mathcal{N}\), where \(\mathcal{N}:=\lbrace 1, 2, \ldots, N-M\rbrace\) is \cite{galleani2008atutorial}
		\begin{equation}
			\begin{aligned}
				\label{e:secondorderclock}
				%  \simode{
					\mat{p_{i}[k+1]\\ f_{i}[k+1]} &= A\mat{p_{i}[k]\\ f_{i}[k]} +B u_i[k] + \mat{v_{1i}[k]\\ v_{2i}[k]},\\
					%  h_i[k]&= \underbrace{\mat{1 && 0}}_{\scriptsize C}\mat{p_{i}[k]\\ f_{i}[k]},
					% }
			\end{aligned}
		\end{equation}
		where \(p_i[k]\) is the time deviation (or phase deviation) of clock \(i\), \(f_i[k]\) is the state related to frequency deviation, and \(u_i[k]\in \R\) is the control input to clock \(i\). 
		The covariance matrix of the white Gaussian process \( \mat{v_{1i}[k] & v_{2i}[k]}^{\top}\) of zero mean is given by
		\vspace{-0.15cm}
		\begin{equation}
			\label{e:covariance second order}
			Q_i = \mat{\tau \sigma_{1i}^2 + \frac{\tau^3}{3}\sigma_{2i}^2 && \frac{\tau^2}{2}\sigma_{2i}^2 \\ \frac{\tau^2}{2}\sigma_{2i}^2 && \tau \sigma_{2i}^2},
		\end{equation} 
		where \(\sigma_{1i}\), \(\sigma_{2i}\) are the standard deviations of white frequency noise and random walk frequency noise of clock \(i\in \mathcal{N}\), respectively.
		
		The discrete-time clock deviation model of a Hm-type atomic clock \(j\in \mathcal{M}\), where \(\mathcal{M}:=\lbrace N-M +1, N-M+2, \ldots, N\rbrace\) is represented by \cite{zucca2005theclock}
		\begin{small}
		\begin{equation}
			\label{e:thirdorderclock}
			\mat{p_{j}[k+1]\\f_{j}[k+1]\\z_{j}[k+1]}=\mat{A & \beta \\ 0 & 1}\mat{p_{j}[k]\\f_{j}[k]\\z_{j}[k]} \\  + \mat{B \\ 0 }u_j[k] + \mat{v_{1j}[k]\\v_{2j}[k]\\v_{3j}[k]},
		\end{equation}
		\end{small}
		where \(p_j[k]\), \(f_j[k]\), and \(z_j[k]\) are the time deviation, state related to frequency deviation, and frequency drift, respectively. The input to clock \(j\) is \(u_j[k]\in \R\).
		The covariance matrix of the white Gaussian process \(\mat{v_{1j}[k] & v_{2j}[k] & v_{3j}[k]}^{\top}\) with zero mean is 
		\begin{small}
			\begin{equation}
				\label{e:covarianceof thirdorderclock}
				Q_j =\mat{
					\tau\sigma_{1j}^2+\frac{\tau^3}{3}\sigma_{2j}^2+\frac{\tau^5}{20}\sigma_{3j}^2 && \frac{\tau^2}{2} \sigma_{2j}^2 + \frac{\tau^4}{8} \sigma_{3j}^2 && \frac{\tau^3}{6} \sigma_{3j}^2\\ \frac{\tau^2}{2}\sigma_{2j}^2 + \frac{\tau^4}{8}\sigma_{3j}^2 && \tau\sigma_{2j}^2 + \frac{\tau^3}{3}\sigma_{3j}^2 && \frac{\tau^2}{2}\sigma_{3j}^2\\
					\frac{\tau^3}{6} \sigma_{3j}^2  && \frac{\tau^2}{2}\sigma_{3j}^2 && \tau \sigma_{3j}^2 },
			\end{equation}
		\end{small}
		where \(\sigma_{1j}\), \(\sigma_{2j}\), and \(\sigma_{3j}\) are standard deviations of the white frequency noise, random walk frequency noise, and random run of frequency drift of clock \(j\in \mathcal{M}\), respectively.
		
		\noindent Using Kronecker product, we can obtain a compact representation of the ensemble as \cite{Dey2024Synchronization}
		\vspace{-0.1cm}
		\begin{equation}
			%\begin{small}
			\label{e:atomicclockensemble}
			\simode{
				\mat{
					x[k+1] \\
					z[k+1]
				}
				&=
				%\underbrace{
					\mat{
						A \otimes I_{N} && \beta \otimes J  \\
						0  && I_M
					}
					%}_{\bm{A}}
				\mat{
					x[k] \\
					z[k]
				}
				+\\ 
				&
				%\underbrace{
					\mat{
						B \otimes I_{N} \\
						0
					}
					%}_{\bm{B}}
				u[k]
				+ 
				\mat{
					v^x [k] \\
					v^{z}[k]
				},\\
				y[k]&= 
				%\underbrace{
					\mat{
						C \otimes V && 0
					}
					%}_{\bm{C}}
				\mat{
					x[k] \\
					z[k]
				}
				+w[k],
			}
			%\end{small}
		\end{equation}
		where \(x[k] :=\mat{p^{\top}[k] & f^{\top}[k]}^{\top}\in \R^{2N}\),
		\vspace{-0.1cm}
		\begin{align*}
			p[k]&:= \mat{p_1[k], \ldots, p_N[k]}^{\top},\; f[k]= \mat{f_1[k], \ldots, f_N[k]}^{\top},\\
			z[k]&:=\mat{z_{N-M+1}[k], \ldots, z_{N}[k]}^{\top}\in \R^{M},\\
			u[k]&:=\mat{u_1[k], u_2[k], \ldots, u_N[k]}^{\top}\in \R^{N},\\
			J&:=\mat{
				0_{(N-M) \times M} \\ I_M
			} \in \mathbb{R}^{N\times M},\; v^{x}[k] := \mat{v^{p}[k]\\ v^{f}[k]},\\
			v^{p}[k]&:=\mat{v_{11}[k], v_{12}[k], \ldots, v_{1N}[k]}^{\top},\\
			v^{f}[k]&:=\mat{v_{21}[k], v_{22}[k], \ldots, v_{2N}[k]}^{\top},\\
			v^z[k]&:=\mat{v_{3 N-M+1}[k], \ldots, v_{3N}[k]}^{\top}.
		\end{align*}
		The measurements \(y[k]\in \R^{N-1}\) are the difference of time deviation (or said differently clock reading) measured independently among the clocks. 
		The matrix \(V\in \R^{(N-1)\times N}\) specifies the pairs of clocks whose clock reading difference is being measured and it is such that \(\kernel V = \image \one_{N}\).
		The measurement noise \(w[k]\) is a white Gaussian noise such that
		\[
		\E\big[w[k]\big] = 0, \,\, \E\big[w[k]w^{\top}[k]\big] = rI_{N-1}\;\; \text{for some}\;\;r\in \R^{+}.
		\]
		
		We assume that all the clocks are independent of each other, then the covariance of the process noise in \eqref{e:atomicclockensemble} is \(Q\in \R^{(2N+M)\times (2N+M)}\) given as
		\begin{equation}
			\label{e:processnoisecovariance}
			Q = \mat{\tau \Sigma_1 + \frac{\tau^3}{3}\Sigma_2 + \frac{\tau^5}{20}\Sigma_3 & \frac{\tau^2}{2}\Sigma_2 + \frac{\tau^4}{8}\Sigma_3 & \frac{\tau^3}{6}\Sigma_4\\ \frac{\tau^2}{2}\Sigma_2 + \frac{\tau^4}{8}\Sigma_3 & \tau \Sigma_2 + \frac{\tau^3}{3}\Sigma_3 & \frac{\tau^2}{2}\Sigma_4 \\ \frac{\tau^3}{6}\Sigma_4^{\top} & \frac{\tau^2}{2}\Sigma_4^{\top} & \tau\Sigma_5}, 
		\end{equation}
		where
		\begin{equation}
			\label{e:covariancecalculation}
			\begin{aligned}
				\Sigma_1&=\diag(\sigma^{2}_{11}, \sigma^{2}_{12},\ldots, \sigma^{2}_{1N}),\\
				\Sigma_2&=\diag(\sigma^{2}_{21}, \sigma^{2}_{22},\ldots, \sigma^{2}_{2N}),\\
				\Sigma_3 &=\diag(\underbrace{0, 0,\ldots, 0}_{N-M}, \sigma^{2}_{3N-M+1},\ldots, \sigma^{2}_{3N}),\\
				\Sigma_4 &= \mat{0 \\ \diag(\sigma^{2}_{3N-M+1}, \ldots, \sigma^{2}_{3N})} \in \R^{N\times M},\\
				\Sigma_5 &= \diag( \sigma^{2}_{3N-M+1},\ldots, \sigma^{2}_{3N}).
			\end{aligned}
		\end{equation}

		As mentioned in the Introduction, CKF continues to be widely applied in practical applications for generating time scales from clock ensembles.
		The CKF algorithm for the mixed ensemble is
		\begin{equation}
			\label{e:CKF}
			\begin{aligned}
				\mat{\hat{x}^{-}[k]\\ \hat{z}^{-}[k]} &= \mathcal{A}\mat{\hat{x}[k-1]\\ \hat{z}[k-1]} + \mathcal{B} u[k-1],\\
				P^{-}[k] &= \mathcal{A}P[k-1]\mathcal{A}^{\top} + Q,\\
				L[k] &= P^{-}[k] \mathcal{C}^{\top} \Big(\mathcal{C}P^{-}[k]\mathcal{C}^{\top} + rI_{N-1}\Big)^{-1},\\
				P[k] &= \Big(I_{2N+ M} - L[k]\mathcal{C}\Big) P^{-}[k],\\
				\mat{\hat{x}[k]\\ \hat{z}[k]} &= \mathcal{A}\mat{\hat{x}^{-}[k]\\ \hat{z}^{-}[k]} + L[k]\Big(y[k]- \mathcal{C}\mat{\hat{x}^{-}[k]\\ \hat{z}^{-}[k]}\Big),
			\end{aligned}
		\end{equation}
		where \(\mathcal{A}= \mat{
			A \otimes I_{N} && \beta \otimes J  \\
			0  && I_M
		}\), \(\mathcal{B}=\mat{
			B \otimes I_{N} \\
			0
		}\), \(\mathcal{C}=\mat{
			C \otimes V && 0
		}\). 
		In \eqref{e:CKF}, \(L[k]\) is the Kalman gain, and \(P^{-}[k]\) and \(P[k]\) are prior and posterior error covariance matrices, respectively.
		
		\noindent In \cite{Dey2024Synchronization}, we observed that the mixed ensemble model \eqref{e:atomicclockensemble} is unobservable with a \(2\)-dimensional unobservable subspace given by \(\image\mat{I_2\otimes \one_{N}\\ 0}\). 
		In fact, the mixed ensemble model is undetectable. 
		Due to undetectability, the error covariance matrix may diverge while running the algorithm.
		This may lead to numerical instability and inaccurate estimation. 
		As a consequence, we address the following problem for the mixed ensemble model \eqref{e:atomicclockensemble}.
		
		\noindent\textbf{Problem 1}: Find a Kalman filtering algorithm for the mixed ensemble \eqref{e:atomicclockensemble} which is free from any diverging terms?
		
		%%%%%%%%%%%%%%%%%%%%%%%%%%%%%%%%%%%%%%%%%%%%%%%%%%%%%%%%%%%%%%%%%%%%%%%%%%%%%%%%%%%%%%%%%%%%%%%%%%%%%%%%%%%%%%%%%%%%%%%%%%%%%%%%%%%
		\subsection{Hadamard Variance (HVAR)}
		\label{ss:hadamard Variance}
		\noindent HVAR is a three-sample variance and can be defined as the variance of the third-order
		difference between successive averages of phase deviation taken over sampling period specified; see \cite{zucca2005theclock, ishizaki2023higher-order}  for more details. 
		For a free-running clock \(i\), the third-order difference is given by%HVAR is used to quantify the fluctuations in frequency/phase of the clock
		\begin{equation*}
			\Delta_{m}^3 p_i[k] = p_i[k+3m] - 3p_i[k+2m] + 3 p_i[k+m] - p_i[k],
		\end{equation*}
		where \(m \geq 1\) is a step. 
		Then, the HVAR is defined as 
		\begin{equation}
			\label{e:hadamard}
			\sigma_{\Hada}^2(\tau) := \E\Big[ \frac{(\Delta_{1}^3 p_i[k])^2}{6\tau^2}\Big].
		\end{equation}
		%where \(\Delta_{1}^3 p_i[k]\) is the third-order difference of \(p_i[k]\). %the clock model is considered as an implicit function of the sampling interval \(\tau\) which can be multiple of \(\tau\), i.e., \(\tau= \ell \tau\).
		In fact, by using \eqref{e:hadamard}, the HVAR of a clock \(i\) can be obtained as \cite{ishizaki2023higher-order}
		\begin{equation}
			\label{e:HadamardFree-running}
			\sigma_{\Hada}^2(\tau) = \frac{1}{\tau} \sigma_{1i}^2 + \frac{\tau}{6}\sigma_{2i}^2 + \frac{11\tau^3}{120} \sigma_{3i}^2,
		\end{equation}
		where it is solely the function of the sampling interval \(\tau\). Note that  \(\sigma_{3i}=0\) for all \(i\in \mathcal{N}\).
		For a controlled clock \(i\), the HVAR is estimated from a clock reading sequence \(\lbrace p_i[0], p_i[1], \ldots, p_i[T]\rbrace\) measured over the sampling interval \(\tau\) as
		%Then, an estimate of HVAR can be computed as
		%\vspace{-0.2cm}
		\begin{equation*}
			\hat{\sigma}^{2}_{\Hada}(m\tau;\lbrace p_i[k]\rbrace):= \frac{1}{T-3m}\sum_{k=0}^{T-3m-1} \frac{(\Delta_m^3 p_i[k])^2}{6(m\tau)^2}.
		\end{equation*}
		%The Hadamard deviation is the square root of the 
		HVAR is represented as a function of the time interval \(m\tau\) using bilogarithmic plots. In this paper, HVAR is employed to assess the generated time scale in the numerical example.
		
		In addition to Problem 1, we consider the problem of generating a time scale from the mixed ensemble model.
		\begin{itemize}
			\item[$\circ$] \textbf{Problem 2(a)}: Generate a time scale by synchronizing the time deviations of all the clocks in the ensemble in the sense
			\begin{equation*}
				\lim_{k\to \infty} \E\big[ p[k]- \mathds{1}_N p^{*}[k]\big] = 0,
			\end{equation*}
			where \(p^{*}[k]\in \R\) is some synchronization state.
			The synchronized time shared by all the clocks is referred as the generated time scale (GTS).
			To improve the accuracy and precision of the GTS, the following two requirements are considered separately.
			\item[$\circ$] \textbf{Problem 2(b)}: Ensure that the GTS has the least short-term HVAR, i.e., we address the following optimization problem
			\begin{equation*}
				\min_{u_i[k]} \hat{\sigma}^{2}_{\Hada}(\tau;\lbrace p_i[k]\rbrace)\quad \forall i\in \mathcal{N} \cup \mathcal{M}.
			\end{equation*}
			\item[$\circ$] \textbf{Problem 2(c)}: We consider the following optimization problem to ensure least long-term HVAR of the GTS
			\begin{equation*}
				\min_{u_i[k]} \hat{\sigma}^{2}_{\Hada}(T;\lbrace p_i[k]\rbrace)\quad \forall i\in \mathcal{N} \cup \mathcal{M},
			\end{equation*}
			where \(T \gg \tau\) denote the longer time interval.
		\end{itemize}
		
		% \eqref{e:atomicclockensemble}.
		% \begin{enumerate}
			% \item Is it possible to modify the Kalman filtering algorithm on the mixed ensemble in such a way that it becomes free from any diverging terms?
			% \item Find a time scale generation algorithm to improve the short-term frequency stability or the long-term frequency stability of the generated time scale?
			% %\item Is it possible to simultaneously improve the short-term and long-term frequency stability of the generated time scale?
			% \end{enumerate}
		
		% %In this work, we also attempt to find techniques to simultaneously improve the short-term and long-term frequency stability of the generated time scale?
		% In this paper, we find answers to the aforementioned problems (1) and (2) by exploiting \emph{ideas from control theory}. 
		% Furthermore, we numerically illustrate that by using the unobservable state estimates, we can simultaneously improve the short-term and long-term frequency stability of the generated time scale.
		%%%%%%%%%%%%%%%%%%%%%%%%%%%%%%%%%%%%%%%%%%%%%%%%%%%%%%%%%%%%%%%%%%%%%%%%%%%%%%%%%%%%%%%%%%%%%%%%%%%%%%%%%%%%%%%%%%%%%%%%%%%%%%%%%%%%
		%%%%%%%%%%%%%%%%%%%%%%%%%%%%%%%%%%%%%%%%%%%%%%%%%%%%%%%%%%%%%%%%%%%%%%%%%%%%%%%%%%%%%%%%%%%%%%%%%%%%%%%%%%%%%%%%%%%%%%%%%%%%%%%%%%%%
		\vspace{-0.2cm}
		\section{Solution of Problem 1}
		\label{s:preliminaryresults}
		%In this section, our aim is to find a Kalman filtering algorithm for the atomic clock ensemble \eqref{e:atomicclockensemble} that does not contain any divergent term. 
		
		\noindent We begin with briefly discussing the observable canonical decomposition employed in \cite{Dey2024Synchronization}, to decompose the system into observable and unobservable states.
		We consider the following transformation
		\begin{equation}
			\label{e:decomposition}
			\mat{x[k]\\z[k]} = \underbrace{\mat{I_2\otimes V^{+} & 0_{2N\times M} & I_2\otimes \one_{N}\\ 0_{M\times 2(N-1)} & I_M & 0_{M\times 2}}}_{\scriptsize \mathbf{T}^{-1}}\mat{\obs[k]\\ \unob[k]},
		\end{equation}
		where \(\obs[k]\in \R^{2(N-1)+M}\) is the observable state, \(\unob[k]\in \R^{2}\) is the unobservable state, \(\mathbf{T}\) is the transformation matrix such that \(\mathbf{T}\mat{x[k]\\z[k]} = \mat{\obs[k]\\ \unob[k]}\). 
		In the transformation \(\mathbf{T}^{-1}\), the matrix \(V^{+}\in \R^{N\times (N-1)}\) is the generalized inverse of \(V\) associated with a weight vector \(q\in \R^{N}\) such that \(q^{\top}\one_N=1\).
		Specifically, we consider a full column rank matrix \(W\in \R^{N\times (N-1)}\) such that \(q^{\top}W = 0\), and we represent the generalized inverse as \(V^{+}= W(VW)^{-1}\). It is easy to see that \(VV^{+}=I_{N-1}\).
		Note that
		\begin{equation}
			\label{e:inversetransform}
			\mathbf{T}=\mat{I_2\otimes V & 0_{2(N-1)\times M}\\ 0_{M\times 2N} & I_M \\ I_2\otimes q^{\top} & 0_{2\times M}}.
		\end{equation}
		
		\noindent Note that \((I_2\otimes q^{\top})x[k]= \unob[k]\) which is the weighted average of clock phase and frequency states.
		Under transformation \(\mathbf{T}\), the new state-space realization is
		\begin{equation}
			\label{e:newstate-space}
			\begin{aligned}
				\mat{\obs[k+1]\\\unob[k+1]} &= \mat{\mathbf{A}_{oo} & 0\\ \mathbf{A}_{\bar{o}o} & A}\mat{\obs[k]\\\unob[k]} + \mat{\mathbf{B}_{o}\\ \mathbf{B}_{\bar{o}}}u[k] \\ &+ \mathbf{T}\mat{v^{x}[k]\\v^{z}[k]},\\
				y[k] &= \mat{\mathbf{C}_{o} & 0_{(N-1)\times 2}}\mat{\obs[k]\\\unob[k]} + w[k],
			\end{aligned}
		\end{equation}
		where 
		\begin{align*}
			\mathbf{A}_{oo} &= \mat{A\otimes I_{N-1} & \beta\otimes VJ \\ 0 & I_M},\\
			\mathbf{A}_{\bar{o}o}&= \mat{0_{2\times 2(N-1)} & \beta\otimes q^{\top}J},\; \mathbf{B}_{o} = \mat{B\otimes V\\ 0},\\
			\mathbf{B}_{\bar{o}} &= B\otimes q^{\top}, \mathbf{C}_{o} = \mat{C\otimes I_{N-1} & 0_{(N-1)\times M}}. 
		\end{align*}
		%%%%%%%%%%%%%%%%%%%%%%%%%%%%%%%%%%%%%%%%%%%%%%%%%%%%%%%%%%%%%%%%%%%%%%%%%%%%%%%%%%%%%%%%%%%%%%%%%%%%%%%%%%%%%%%%%%%%%%%%%%%%%%%%%%%
		%\vspace{-0.8cm}
		\subsection{Transformed Kalman Filter (TKF)}
		\label{ss:SS_TKF}
		
		\noindent In this subsection, we obtain the Kalman filtering algorithm for the transformed state-space realization \eqref{e:newstate-space} which does not have any diverging terms.
		For the new state-space realization, the steps of the Kalman filtering algorithm are
		%\begin{equation}
		%\label{e:newkalmafilterequations}
		\begin{small}
		\begin{align}
			\label{e:prior}
			\mat{\obspre[k]\\\unobpre[k]} &=  \mat{\mathbf{A}_{oo} & 0\\ \mathbf{A}_{\bar{o}o} & A}\mat{\obsest[k-1]\\\unobest[k-1]} + \mat{\mathbf{B}_{o}\\ \mathbf{B}_{\bar{o}}}u[k-1],\\
			\label{e:priorcovold}
			\breve{P}^{-}[k] &=  \mat{\mathbf{A}_{oo} & 0\\ \mathbf{A}_{\bar{o}o} & A} \breve{P}[k-1]\mat{\mathbf{A}_{oo} & 0\\ \mathbf{A}_{\bar{o}o} & A}^{\top} +\mathbf{T}Q\mathbf{T}^{\top},
		\end{align}
		\begin{equation}
			\begin{aligned}
				\label{e:kalmangainold}
				\mat{\obsgain[k]\\ \unobgain[k]} &= \breve{P}^{-}[k] \mat{\mathbf{C}_{o} & 0}^{\top}\Big(\mat{\mathbf{C}_{o} & 0}\breve{P}^{-}[k]\mat{\mathbf{C}_{o} & 0}^{\top}\\ & \hspace{4.5 cm} + rI_{N-1}\Big)^{-1},
			\end{aligned}
		\end{equation} 
		\begin{align}
			\label{e:postcovold}
			\breve{P}[k] &= \Big(I_{2N+M} - \mat{\obsgain[k]\\\unobgain[k]}\mat{\mathbf{C}_{o} & 0}\Big)\breve{P}^{-}[k]\\
			\label{e:post}
			\mat{\obsest[k]\\\unobest[k]} &= \mat{\obspre[k]\\\unobpre[k]} + 
\mat{\obsgain[k]\\\unobgain[k]}\Big(y[k] - \mat{\mathbf{C}_{o} & 0} \mat{\obspre[k]\\\unobpre[k]}\Big).
		\end{align}
		\end{small}
		%\end{equation}
		
		\noindent For further analysis, we denote \(\breve{P}^{-}[k]\) and \(\breve{P}[k]\) as
		\begin{equation*}
			\breve{P}^{-}[k] = \mat{P^{-}_{oo}[k] & (P^{-}_{\bar{o}o}[k])^{\top}\\ P^{-}_{\bar{o}o}[k] & P^{-}_{\bar{o}\bar{o}}[k]},\breve{P}[k] = \mat{P_{oo}[k] & P_{\bar{o}o}^{\top}[k]\\ P_{\bar{o}o}[k] & P_{\bar{o}\bar{o}}[k]} , 
		\end{equation*}
		where \(P^{-}_{oo}[k]\) and \(P_{oo}[k]\) are error covariance matrices associated with observable state, and \(P^{-}_{\bar{o}\bar{o}}[k]\) and \(P_{\bar{o}\bar{o}}[k]\) are error covariance matrices associated to unobservable state. \(P^{-}_{\bar{o}o}[k]\) and \(P_{\bar{o}{o}}[k]\) represent the estimation errors between the unobservable and observable state. \(\obsgain[k]\) is the observable gain and \(\unobgain[k]\) is the unobservable gain.

		Then, the following lemma shows that the computation of Kalman gains \(\obsgain[k]\) and \(\unobgain[k]\) depends only on \(P^{-}_{oo}[k]\) and \(P^{-}_{\bar{o}o}[k]\).
		
		\begin{lemma}
			\label{l:kalmanfilterformixedensemble}
			Consider the mixed ensemble \eqref{e:atomicclockensemble}, the observable canonical decomposition \eqref{e:decomposition}, and the transformed state-space realization \eqref{e:newstate-space}. Then the Kalman filtering algorithm for the transformed state-space is given as 
			%   \begin{equation}
				%\label{e:obtainedkalmafilterequations}
			
				\begin{align}
					\mat{\obspre[k]\\\unobpre[k]} &=  \mat{\mathbf{A}_{oo} & 0\\ \mathbf{A}_{\bar{o}o} & A}\mat{\obsest[k-1]\\\unobest[k-1]} + \mat{\mathbf{B}_{o}\\ \mathbf{B}_{\bar{o}}}u[k-1],\\
					\label{e:priorcovnew}
					\mat{\errorobspre[k]\\\errorunobpre[k]} &=  \mat{\mathbf{A}_{oo} & 0\\ \mathbf{A}_{\bar{o}o} & A}  \mat{\errorobsest[k-1]\\\errorunobest[k-1]}\mathbf{A}^{\top}_{oo} +\mat{\mathbf{T}_o Q \mathbf{T}_o^{\top}\\ \mathbf{T}_{\bar{o}} Q\mathbf{T}_{o}^{\top}},\\
					\label{e:kalmangainnew}
					\mat{\obsgain[k]\\\unobgain[k]} &=  \mat{\errorobspre[k]\\\errorunobpre[k]}\mathbf{C}^{\top}_{o}\Big(\mathbf{C}_{o}\errorobspre[k]\mathbf{C}_{o}^{\top} + rI_{N-1}\Big)^{-1},\\
					\label{e:postcovnew}
					\mat{\errorobsest[k]\\\errorunobest[k]} &= \mat{\Big(I_{2(N-1)+M} - \obsgain[k]\mathbf{C}_{o}\Big)\errorobspre[k]\\ \errorunobpre[k]\Big(I_{2(N-1) + M}- \mathbf{C}^{\top}_{o}\obsgain^{\top}[k]\Big)}\\
					\mat{\obsest[k]\\\unobest[k]} &= \mat{\obspre[k]\\\unobpre[k]} +\mat{\obsgain[k]\\\unobgain[k]}\Big(y[k] - \mathbf{C}_{o} {\obspre[k]}\Big),
				\end{align}
		
				where
				\[
				\mathbf{T}_{o} = \mat{I_2\otimes V & 0 \\ 0 & I_M},\; \mathbf{T}_{\bar{o}} = \mat{I_2\otimes \weight^{\top} & 0}.
				\]
				%\end{equation}
			\end{lemma}
			\vspace{0.5cm}
			\begin{proof}
				In \eqref{e:priorcovold}, if we further analyze then we observe that the update of \(\errorobspre[k]\) and \(\errorunobpre[k]\) are independent of \(P_{\bar{o}\bar{o}}[k]\), which leads to \eqref{e:priorcovnew}.
				
				Note that
				\begin{equation*}
					\mat{P^{-}_{oo}[k] & (P^{-}_{\bar{o}o}[k])^{\top}\\ P^{-}_{\bar{o}o}[k] & P^{-}_{\bar{o}\bar{o}}[k]}\mat{\mathbf{C}^{\top}_o\\ 0_{2\times (N-1)}} = \mat{P^{-}_{oo}[k]\\ P^{-}_{\bar{o}o}[k]}\mathbf{C}^{\top}_o.
				\end{equation*}
				This shows that the Kalman gains \(\obsgain[k]\) and \(\unobgain[k]\) does not depend on \(P_{\bar{o}\bar{o}}[k]\), and result in \eqref{e:kalmangainnew}.
				On further calculation in \eqref{e:postcovold}, we get \(\errorunobest^{\top}[k]= (I_{2(N-1)+M}- \obsgain[k]\mathbf{C}_o){\errorunobpre}^{\top}[k]\), which leads to \eqref{e:postcovnew}.
			\end{proof}
			\noindent Note that \(P^{-}_{\bar{o}\bar{o}}[k]\) and \(P_{\bar{o}\bar{o}}[k]\) are responsible for the divergence of the error covariance matrices, and can be removed from the Kalman filtering algorithm for the mixed ensemble \eqref{e:atomicclockensemble} as shown in \eqref{e:priorcovnew} and \eqref{e:postcovnew}.
			%In this work, the estimates of the unobservable state would be beneficial to improve the frequency stability of the time scale generated.
			
			Although the error covariance matrices of CKF \eqref{e:CKF} diverges, the error covariance matrices \(P^{-}_{oo}[k]\) and \(P^{-}_{\bar{o}o}[k]\) converges to its steady-state value\cite{Burridge1987convergence}. 
			Regarding the observable state, we denote the steady-state error covariance and steady-state observable gain as
			\begin{equation}
				\label{e:steadystateobservable gain}
				\begin{aligned}
					\lim_{k\to \infty} P^{-}_{oo}[k]& = P^{*}_{oo},\\ \lim_{k\to \infty} P_{oo}[k] &= \Big(I_{2(N-1)+ M}- L_o^{*}\mathbf{C}_{o}\Big)P^{*}_{oo}
				\end{aligned}
			\end{equation}
			\begin{equation}
				\label{e:steady-stateobservablegain}
				\lim_{k\to \infty} L_{o}[k]= L^{*}_{o}=P^{*}_{oo}\mathbf{C}^{\top}_{o}\Big(\mathbf{C}_oP^{*}_{oo}\mathbf{C}^{\top}_{o} + rI_{N-1}\Big)^{-1}, 
			\end{equation}
			where \(P_{oo}^{*}\) is the unique solution to the Riccatti equation
			\begin{equation*}
				\begin{aligned}
					P^{*}_{oo} = \mathbf{T}_{o}&Q\mathbf{T}^{\top}_{o}+ \mathbf{A}_{oo}P_{oo}^{*}\mathbf{A}_{oo}^{\top} \\- &\mathbf{A}_{oo}P_{oo}^{*}\mathbf{C}^{\top}_{o}\Big(\mathbf{C}_oP^{*}_{oo}\mathbf{C}^{\top}_{o} + rI_{N-1}\Big)^{-1}\mathbf{C}_{o}P_{oo}^{*}\mathbf{A}_{oo}^{\top}. 
				\end{aligned}
			\end{equation*}
			
			For the calculation of the steady-state unobservable Kalman gain, we consider
			\begin{equation}
				\label{e:steadystateunpriorcovariance}
				\begin{aligned}
					\lim_{k\to \infty} P^{-}_{\bar{o}o}[k]& = P^{*}_{\bar{o}o}\\
				\end{aligned}
			\end{equation}
			By \eqref{e:priorcovnew} and \eqref{e:postcovnew}, it holds that
			\begin{equation}
				\label{e:steadystateunobservablecovarianceintermediate}
				\begin{aligned}
					P^{*}_{\bar{o}o}=  \mathbf{T}_{\bar{o}}Q\mathbf{T}^{\top}_{o}& + AP^{*}_{\bar{o}o}\mathbf{A}^{\top}_{oo} - AP^{*}_{\bar{o}o} \mathbf{C}^{\top}_{o}(L^{*}_{o})^{\top}\mathbf{A}^{\top}_{oo} \\ + \mathbf{A}_{\bar{o}o} P_{oo}^{*}\mathbf{A}^{\top}_{oo} &- \mathbf{A}_{\bar{o}o}L^{*}_{o}\mathbf{C}_{o}  P^{*}_{oo}\mathbf{A}^{\top}_{oo}.
				\end{aligned}
			\end{equation}
			The steady-state unobservable Kalman gain \(L^{*}_{\bar{o}}\) is given by
			\begin{equation}
				\label{e:steady-stateunobservablegain}
				L^{*}_{\bar{o}} = P^{*}_{\bar{o}o} \mathbf{C}^{\top}_{o}\Big(\mathbf{C}_oP^{*}_{oo}\mathbf{C}^{\top}_{o} + rI_{N-1}\Big)^{-1}.
			\end{equation}
			
			\noindent Then the steady-state Kalman filter for the mixed ensemble \eqref{e:atomicclockensemble} in the transformed state-space realization is as follows.
			\begin{theorem}
				\label{t:steady-stateKalman filter}
				Consider the mixed ensemble model \eqref{e:atomicclockensemble}, the observable canonical decomposition \eqref{e:decomposition}, transformed state-space realization \eqref{e:newstate-space}. Then the steady-state Kalman filtering algorithm for the transformed state-space is given as 
				\begin{equation}
					\label{e:transformedsteadystateKF}
					\begin{aligned}
						\mat{\obspre[k]\\\unobpre[k]} &=  \mat{\mathbf{A}_{oo} & 0\\ \mathbf{A}_{\bar{o}o} & A}\mat{\obsest[k-1]\\\unobest[k-1]} + \mat{\mathbf{B}_{o}\\ \mathbf{B}_{\bar{o}}}u[k-1],\\
						\mat{\obsest[k]\\\unobest[k]} &= \mat{\obspre[k]\\\unobpre[k]} +\mat{\obsgain^{*}\\\unobgain^{*}}\Big(y[k] - \mathbf{C}_{o} {\obspre[k]}\Big),
					\end{aligned}
				\end{equation}
				where the observable and unobservable Kalman gains are given in \eqref{e:steady-stateobservablegain} and \eqref{e:steady-stateunobservablegain}, respectively. 
			\end{theorem}
			\begin{proof}
				For the observable state, the steady-state Kalman filtering algorithm is the same as that of the conventional Kalman filtering algorithm on the observable state.
				
				For the unobservable part, we prove that \(\errorunobest^{*}\) is uniquely determined.
				We can write \(\errorunobest^{*}\) in \eqref{e:steadystateunobservablecovarianceintermediate} as
				\begin{equation}
					\begin{aligned}
						\label{e:above}
						\errorunobest^{*} = \mathbf{T}_{o}Q\mathbf{T}^{\top}_{o} &+ A\errorunobest^{*}\Big(I_{2(N-1)+M} - L^{*}_{o}\mathbf{C}_o\Big)^{\top}\mathbf{A}_{oo}^{\top}\\ &+ \mathbf{A}_{\bar{o}o}\Big(I_{2(N-1)+M} - L^{*}_{o}\mathbf{C}_o\Big)\errorobsest^{*} \mathbf{A}^{\top}_{oo}  
					\end{aligned}
				\end{equation}
				By applying vectorization of the matrices, we can obtain column expansion of \eqref{e:above} as
				\begin{align*}
					\vectorize(\errorunobest^{*}) = \Big ( \mathbf{A}_{oo}&(I_{2(N-1)+M} - L^{*}_{o}\mathbf{C}_o)\otimes A \Big ) \vectorize(\errorunobest^{*}) \\ 
					+ \vectorize\Big(\mathbf{T}_{o}Q\mathbf{T}^{\top}_{o} & + \mathbf{A}_{\bar{o}o}\Big(I_{2(N-1)+M} - L^{*}_{o}\mathbf{C}_o \Big) \errorobsest^{*}\mathbf{A}^{\top}_{oo}\Big).
				\end{align*}
				
				\noindent The eigenvalues of the matrix \(\mathbf{A}_{oo}(I-L^{*}_{o}\mathbf{C}_o)\otimes A\) is given by the multiplication of the eigenvalues of 
				the matrix \(\mathbf{A}_{oo}(I-L^{*}_{o}\mathbf{C}_o)\) and eigenvalues of \(A\). 
				Since \(L^{*}_{o}\) is the steady-state Kalman gain of the observable part, each eigenvalue of \(\mathbf{A}_{oo}(I-L^{*}_{o}\mathbf{C}_o)\), \(\lambda\) satisfies \(\abs{\lambda} < 1\). Note that all the eigenvalues of \(A\) are one. 
				As a consequence, all the eigenvalues of \(\mathbf{A}_{oo}(I-L^{*}_{o}\mathbf{C}_o)\otimes A\) are within the unit circle.
				This confirms the uniqueness of \(\errorunobest^{*}\).
			\end{proof}
			
			%We call the obtained steady-state Kalman filtering algorithm for the mixed ensemble model \eqref{e:atomicclockensemble} as steady-state transformed Kalman filtering algorithm \eqref{e:transformedsteadystateKF}.
			Hence, we have obtained the solution to Problem 1.
			\begin{remark}
				\label{r:note on SS-TKF}
				In \cite{Ishizaki2025Explicit}, a Kalman filtering algorithm which is free from any diverging terms is obtained for an atomic clock ensemble consisting of only Cs-type clocks. 
				The transformed Kalman filtering algorithm, we obtain, is a generalization of the algorithm given in \cite{Ishizaki2025Explicit}.
				However, the steps employed for obtaining the TKF and its steady-state form in both works are similar.
			\end{remark}

			%%%%%%%%%%%%%%%%%%%%%%%%%%%%%%%%%%%%%%%%%%%%%%%%%%%%%%%%%%%%%%%%%%%%%%%%%%%%%%%%%%%%%%%%%%%%%%%%%%%%%%%%%%%%%%%%%%%%%%%%%%%%%%%%%%%
			%\vspace{-0.2cm}
			\section{Solution of Problem 2}
			\label{s:solutionofProb2}
			\subsection{Explicit Ensemble Mean Synchronization Algorithm}
			\label{s:ensemble mean snchronization algorithm}
			\noindent In this subsection, we provide details of the explicit ensemble mean synchronization algorithm for generating a time scale.
			In particular, to achieve synchronization, we consider the control input as
			\begin{equation}
				\label{e:controlinput}
				u[k] = V^{+} \varphi[k].% + \one_{N}\bar{\omega}[k].
			\end{equation}
			
			To efficiently design the control input \(\varphi[k]\) in \eqref{e:controlinput}, we employ the idea of combining state-feedback control with steady-state transformed Kalman filter \eqref{e:transformedsteadystateKF}.% obtained in Subsection \ref{ss:SS_TKF}.
			
			Based on the observable state estimates, we design the control \(\varphi[k]\) for synchronization as
			\begin{equation}
				\label{e:synccontrol}
				\varphi[k] = -\mat{F\otimes I_{N-1} & D\otimes VJ}\obsest[k],
			\end{equation}
			where \(\obsest[k]\) is computed from the steady-state transformed Kalman filter via \eqref{e:transformedsteadystateKF} and the feedback gains \(F\in \R^{1\times 2}\) and \(D\in \R\) are to be chosen later. 
			% \begin{equation}
				%     \obspre[k] = \mathbf{A}_{oo} \obspre[k-1] + \mathbf{B}_{o}u[k-1] + \mathbf{A}_{oo}\obsgain^{*} (y[k] - \mathbf{C}_{o}\obspre[k]),
				% \end{equation}
			
			Before we state our synchronization result, we have the following remark.
			\begin{remark}
				\label{r:effectofdriftonunobservable}
				Let us closely look at the dynamics of the observable state \(\obs[k]\in \R^{2(N-1)+ M}\). 
				Suppose we partition \(\obs[k]=\mat{\xi_{o}[k]\\ \nu_o[k]}\) where \(\xi_{o}[k]\in \R^{2(N-1)}\) and \(\nu_o[k]\in \R^{M}\).
				Then, on the basis of the observable canonical decomposition \eqref{e:decomposition}, we see that \(\nu_o[k] = z[k]\) which is a column vector containing the frequency drift state of all the Hm-type clocks.
				Then, the unobservable state dynamics is
				\begin{equation}
					\label{e:modified unobservable state}
					\begin{aligned}
						\unob[k+1] &= A\unob[k] + (\beta\otimes q^{\top}J)\nu_o[k] \\ &+ (I_2\otimes q^{\top})v^{x}[k] + (B\otimes q^{\top})u[k],
					\end{aligned}
				\end{equation}
				where the unobservable state is effected by the observable state via \(\nu_o[k]\).
			\end{remark}
			
			Then the following result gives the solution of Problem 2(a).
			
			\begin{theorem}
				\label{t:synctheorem}
				Consider the mixed ensemble \eqref{e:atomicclockensemble} and the observable canonical decomposition \eqref{e:decomposition}. 
				Let the input \(u[k]\) is taken as \eqref{e:controlinput} with \(\varphi[k]\) given in \eqref{e:synccontrol}. 
				Let the feedback gain matrices are structured as
				\begin{equation}
					F = \mat{\frac{\gamma}{\tau}  & 1},\quad D=\frac{\tau}{2},
				\end{equation}
				where \(\gamma\) is a scalar parameter. Then
				\begin{equation}
					\label{e:mainsyncresult}
					\lim_{k\to \infty} \E\Big[p[k] - \one_{N} \theta_{\bar{o}}[k]\Big] = 0
				\end{equation}
				if and only if \(\abs{1-\gamma} < 1\), where the dynamics of the synchronization destination (unobservable state) is
				\begin{equation}
					\label{e:syncstate}
					\simode{
						\unob[k+1] &= A\unob[k] + (\beta\otimes q^{\top}J)\nu_o[k] + (I_2\otimes q^{\top})v^{x}[k],\\
						\theta_{\bar{o}}[k] &= C\unob[k].
					}
				\end{equation}
			\end{theorem}
			\begin{proof}
				Note that \(q^{\top}V^{+}=0\). Then \((B\otimes q^{\top})u[k]=0\) and the dynamics of the unobservable state becomes as in \eqref{e:syncstate}.
				Recall that  \(\obs[k]=\mat{\xi_{o}[k]\\ \nu_o[k]}\) where \(\xi_{o}[k]\in \R^{2(N-1)}\) and \(\nu_o[k]\in \R^{M}\).
				For our analysis, we can further partition \(\xi_{o}[k]\) into  \(\xi_o[k] = \mat{\xi^{p}_{o}[k] \\ \xi^{f}_{o}[k]}\), where \(\xi^{p}_{o}[k]\in \R^{N-1}\) and \(\xi^{f}_{o}[k]\in \R^{N-1}\). By the observable canonical decomposition \eqref{e:decomposition}, we see that \(p[k] - \one_{N} \theta_{\bar{o}}[k] = V^{+}\xi^{p}_{o}[k]\).
				% \begin{equation*}
					%     p[k] - \one_{N} \theta_{\bar{o}}[k] = V^{+}\xi^{p}_{o}[k]
					% \end{equation*}
				Thus, to prove \eqref{e:mainsyncresult}, we have to show that \(\E[\xi^{p}_{o}[k]] = 0\) and \(\xi^{p}_{o}[k]\) has bounded variance.
				
				The dynamics of the observable state  \(\obs[k]=\mat{\xi_{o}[k]\\ \nu_o[k]}\) under the input \(u[k]\) is 
				\begin{equation}
					\begin{aligned}
						\mat{\xi_{o}[k+1]\\ \nu_o[k+1]} &= \mat{A\otimes I_{N-1} & \beta\otimes VJ \\ 0_{M\times 2(N-1)} & I_M}\mat{\xi_{o}[k]\\ \nu_o[k]}\\ & + \mat{B\otimes V\\ 0_{M\times N}}u[k]+ \mat{(I_2\otimes V)v^{x}[k]\\ v^{z}[k]}\\
						&= \mat{A\otimes I_{N-1} & \beta\otimes VJ \\ 0 & I_M}\mat{\xi_{o}[k] \\ \nu_{o}[k]} \\ -\mat{B\otimes I_{N-1}\\ 0_{M\times (N-1)}}&\mat{F\otimes I_{N-1} & D\otimes VJ}\mat{\hat{\xi}_{o}[k]\\ \hat{\nu}_o[k]} \\ &+ \mat{(I_2\otimes V)v^{x}[k]\\ v^{z}[k]}
					\end{aligned}
				\end{equation}
				Let the state estimation error be
				\begin{equation*}
					\mat{e_{\xi_{o}}[k]\\e_{\nu_{o}}[k]}=\mat{\xi_{o}[k]\\\nu_{o}[k]}-\mat{\hat{\xi}_{o}[k]\\\hat{\nu} _{o}[k]}.\;\text{Then} 
				\end{equation*}
				% Then 
				\begin{small}
					\begin{equation*}
						\begin{aligned}
							\mat{\xi_{o}[k+1]\\\nu_o[k+1]}=\mat{(A-BF)\otimes I_{N-1} & (\beta-DB)\otimes VJ\\ 0 & I_M}\mat{\xi_{o}[k]\\\nu_{o}[k]} \\+\mat{BF\otimes I_{N-1} & DB\otimes VJ\\ 0 & 0} \mat{e_{\xi_{o}}[k]\\e_{\nu_{o}}[k]} + \mat{(I_2\otimes V)v^x[k] \\ v^z[k]}. 
						\end{aligned}
					\end{equation*}
				\end{small}   
				We see that 
				\begin{equation*}
					A-BF=\mat{1 & \tau\\ 0 & 1}-\mat{\tau\\ 1}\mat{\frac{\gamma}{\tau}& 1}=\mat{1-\gamma & 0 \\ -\frac{\gamma}{\tau} & 0}.   
				\end{equation*}
				\begin{equation*}
					\beta-DB=\mat{\frac{\tau^2}{2}\\\tau}-\frac{\tau}{2}\mat{\tau\\ 1}=\mat{0 \\ \frac{\tau}{2}}.\;\text{Then}
				\end{equation*}
				%Then 
				\begin{equation}
					\label{e:intermediateproof}
					\begin{aligned}
						\xi_{o}[k+1]&=\Big(\mat{1-\gamma & 0 \\ -\frac{\gamma}{\tau} & 0}\otimes I_{N-1}\Big){\xi_{o}[k]}\\&+ 
						\Big(\mat{0\\ \frac{\tau}{2}}\otimes VJ\Big)\nu_o[k]+\Big(\mat{\gamma & \tau\\ 
							\frac{\gamma}{\tau} & 1}\otimes I_{N-1}\Big)e_{\xi_{o}}[k] \\
						&+ \Big(\mat{\frac{\tau^2}{2}\\ \frac{\tau}{2}}\otimes VJ\Big)e_{\nu_{o}}[k] + (I_2\otimes V)v^x[k].
					\end{aligned}
				\end{equation}
				From \eqref{e:intermediateproof}, the dynamics of \(\xi^{p}_{o}[k]\) is derived as
				\begin{equation*}
					\begin{aligned}
						\xi^{p}_{o}[k+1]&=(1-\gamma)I_{N-1}\xi^{p}_{o}[k] + \mat{\gamma I_{N-1} & \tau I_{N-1}}{e_{\xi_{o}}[k]} \\ &+\frac{\tau^2}{2}VJ e_{\nu_{o}}[k] + V v^p[k].  
					\end{aligned}
				\end{equation*}
				Since, we consider Kalman filter for the observable state estimates, the errors \({e_{\xi_{o}}[k]}\), and  \(e_{\nu_{o}}[k]\) are convergent.
				The eigenvalues of \((1-\gamma)I_{N-1}\) are within the unit circle if and only if \(\abs{1-\gamma}\;<\;1\). Thus, \(\lim_{k\to \infty}\E[\xi^{p}_{o}[k]]=0\). 
				Since the error and noise terms in the dynamics of \(\xi^{p}_{o}[k]\) have finite variance, the variance of \(\xi^{p}_{o}[k]\) is also bounded.
				This completes the proof.
			\end{proof}
			% \begin{remark}
				%  \label{r:differencefromCDC2024}  
				%  In our previous work \cite{Dey2024Synchronization}, we assume that the weight vector \(q\) that zero weights are assigned to all the Hm-type clocks, i.e., \(q_{N-M+1}= q_{N-M+2}=\ldots = q_{N}=0\). 
				%  however, in this work, we do not impose any assumption on \(q\). 
				%  In fact, in the next subsection, we chose the weight vector appropriately to improve the frequency stability of generated time scale.
				%  In contrast to CDC2024, we obtained an if and only if condition in Theorem \ref{t:synctheorem}.
				% \end{remark}
			To tackle Problems 2(b) and 2(c), in addition to the synchronization achieved in Theorem \ref{t:synctheorem}, it is necessary to maximize the frequency stability of the GTS.
			To achieve this, we consider the following free-running dynamics 
			\begin{equation}
				\label{e:free-runningensemble}
				\Psi(q) := \simode{g[k+1] &= A g[k] + (I_2\otimes q^{\top})v^{x}[k],\\
					h_{g}[k] &= C g[k].}
			\end{equation}
			
			\noindent In the following subsection, we evaluate the HVAR of the free-running dynamics \(\Psi(q)\) to suitably choose the weight vector \(q\) in order to optimize the short-term or long-term frequency stability of the resulting time scale.
			%%%%%%%%%%%%%%%%%%%%%%%%%%%%%%%%%%%%%%%%%%%%%%%%%%%%%%%%%%%%%%%%%%%%%%%%%%%%%%%%%%%%%%%%%%%%%%%%%%%%%%%%%%%%%%%%%%%%%%%%%%%%%%%%%%%
			\subsubsection{Hadamard Variance of the Free-running Dynamics}
			\label{ss:Hadamardvariancesyncdestination}
			\noindent We aim to find an optimal weight vector which minimizes the Hadamard variance of the free-running dynamics \(\Psi(q)\), denoted by \(\sigma^{2}_{\Hada}(\tau; \Psi(q))\), and it is defined as
			\begin{equation*}
				\sigma^{2}_{\Hada}(\tau; \Psi(q)) :=  \E\Big[ \frac{(\Delta_{1}^3 h_g[k])^2}{6\tau^2}\Big].    
			\end{equation*}
			By further computation, we get
			\begin{equation*}
				\Delta_{1}^{3}h_g[k] =X\mat{(I_2\otimes q^{\top})v^{x}[k]\\ (I_2\otimes q^{\top})v^{x}[k+1]\\ (I_2\otimes q^{\top})v^{x}[k+2]},
			\end{equation*}
			where \(X=\mat{C(A^2-3A+3I_2) & C(A-3I_2) & C}\). 
			Then, the Hadamard variance is computed as 
			\begin{equation*}
				\label{e:intermediateHadarmardcalculation}
				\sigma_{\Hada}^2(\tau; \Psi(q)) = \frac{1}{6\tau^2}X \Big\lbrace I_3 \otimes [(I_2\otimes q^{\top}){Q}^{x}(I_2\otimes q)]\Big\rbrace X^{\top}, 
			\end{equation*}
			where \(Q^{x}\) is given by
			\begin{equation*}
				{Q}^{x} = \mat{\tau \Sigma_1 + \frac{\tau^3}{3}\Sigma_2 + \frac{\tau^5}{20}\Sigma_3 & \frac{\tau^2}{2}\Sigma_2 + \frac{\tau^4}{8}\Sigma_3 \\ \frac{\tau^2}{2}\Sigma_2 + \frac{\tau^4}{8}\Sigma_3 & \tau \Sigma_2 + \frac{\tau^3}{3}\Sigma_3}.
			\end{equation*}
			with \(\Sigma_1\), \(\Sigma_2\), and \(\Sigma_3\) given in \eqref{e:covariancecalculation}.
			% \begin{equation*}
				% \begin{aligned}
					% \Sigma_1&=\diag(\sigma^{2}_{11}, \sigma^{2}_{12},\ldots, \sigma^{2}_{1N}),\\
					% \Sigma_2&=\diag(\sigma^{2}_{21}, \sigma^{2}_{22},\ldots, \sigma^{2}_{2N}),\\
					% \Sigma_3 &=\diag(\underbrace{0, 0,\ldots, 0}_{N-M}, \sigma^{2}_{3N-M+1},\ldots, \sigma^{2}_{3N}).
					% \end{aligned}
				% \end{equation*}
			
			On further calculations, we get
			\begin{equation}
				\label{e:Hadmaardcalculationmixed}
				\sigma_{\Hada}^2(\tau; \Psi(q)) = \frac{q^{\top}\Pi(\tau)q}{\tau^2},
			\end{equation}
			where \(\Pi(\tau)= \tau \Sigma_1 + \frac{\tau^3}{6}\Sigma_2 + \frac{13\tau^5}{360}\Sigma_3\).
			
			\noindent Based on the above calculations, the optimal weight vector \(q\) is obtained by solving the following optimization problem
			\begin{equation}
				\label{e:optimalweight}
				\begin{aligned}
					q_{\Hada}(\tau):= &\argmin_{q}  \sigma_{\Hada}^2(\tau;\Psi(q))\;\; \suth \;\;q^{\top}\one_{N}=1,
				\end{aligned}
			\end{equation}
			
			The solution of the optimization problem \eqref{e:optimalweight} is given by the following result.
			
			\begin{theorem}
				\label{t:optimalweight}
				Consider the free-running dynamics \(\Psi(q)\). Then the weight vector that solves the optimization problem \eqref{e:optimalweight} is given by
				\begin{equation}
					\label{e:optimalweightresult}
					q_{\Hada}(\tau)= \frac{\Pi^{-1}(\tau) \one_{N}}{\one_{N}^{\top}\Pi^{-1}(\tau) \one_{N}}.  
				\end{equation}
			\end{theorem}
			\begin{proof} 
				By Lagrangian multiplier method, the optimum of \eqref{e:optimalweight} is found as the solution of 
				\begin{equation*}
					2(\tau^2)^{-1}\Pi(\tau)q - \lambda\one_{N} = 0, \quad q^{\top}\one_{N}=1,
				\end{equation*}
				which results in \eqref{e:optimalweightresult}. This completes the proof.
			\end{proof}
			
			\noindent Based on Theorem \ref{t:optimalweight}, we can find the weight vector to improve the short-term or long-term frequency stability of the free-running dynamics. 
			The weight vector which maximizes the short-term frequency stability can be found as
			\begin{equation}
				\label{e:shorttermweight}
				q_{\Hada}^{0}:= \lim_{\tau \to 0} q_{\Hada}(\tau) = \frac{\Sigma_1^{-1}\one_{N}}{\one^{\top}_{N}\Sigma_1^{-1}\one_{N}},
			\end{equation}
			where \(\Sigma_1 =\diag(\sigma^2_{11},\sigma^2_{12}, \ldots, \sigma^2_{1N})\).% and \(\sigma_{1i}\) is the standard deviation of the white frequency noise of clock \(i\) in the ensemble.
			
			%\begin{remark}(On short-term stability):
			%\label{r:short-term}
			Observe that the dynamics of the synchronization destination in \eqref{e:syncstate} is effected by \(\nu_{o}[k]\). 
			%Note that \(\nu[k]= z[k]\). 
			By \eqref{e:HadamardFree-running}, we can infer that the frequency drift of a clock has negligible effect in shorter intervals as compared to longer intervals.
			As a consequence, if we set the weight vector \(q= q_{\Hada}^{0}\) then the behavior of the unobservable state \emph{approximately} coincides with \(\Psi(q_{\Hada}^{0})\). Since \(p[k]\) follows \(\theta_{o}[k]\) in expectation as shown in Theorem \ref{t:synctheorem}, by setting \(q= q_{\Hada}^{0}\), we can obtain a time scale with least short-term HVAR. This solves Problem 2(b).% and \eqref{e:shorttermweight}, this results in achieving the best short-term stability of the generated time scale in shorter period of time.   
			%\end{remark}
			%Theorem \ref{t:synctheorem}, Theorem \ref{t:optimalweight}, and Remark \ref{r:long-term} provide solution to generating a time scale with excellent short-term stability.
			
			In the similar manner, the weight vector that maximizes the long-term frequency stability of \(\Psi(q)\) is obtained as
			\begin{equation}
				\label{e:longtermweight}
				q_{\Hada}^{\infty}:= \lim_{\tau \to \infty}  q_{\Hada}(\tau) = \mat{\frac{\Sigma_{2, N-M}^{-1}\one_{N-M}}{\one^{\top}_{N-M}\Sigma_{2,N-M}^{-1}\one_{N-M}}\\ 0_{M\times 1}},
			\end{equation}
			where \(\Sigma_{2,N-M}=\diag(\sigma^2_{21},\sigma^2_{22}, \ldots, \sigma^2_{2N-M})\) is a submatrix of \(\Sigma_2\) whose diagonal elements contain the standard deviations of the random walk frequency noise of the Cs-type clocks.
			Intuitively, it is reasonable to associate zero weight to each Hm-type clock since they have lower frequency stability in longer intervals as compared to Cs-type clocks.
			%\begin{remark}(On long-term stability):
			%\label{r:long-term}
			
			It is interesting to note that if \(q\) is set equal to \(q_{\Hada}^{\infty}\) then \((q_{\Hada}^{\infty})^{\top}J = 0\).
			As a consequence, the dynamics of the unobservable state \eqref{e:syncstate} becomes
			\[
			\unob[k+1] = A\unob[k] + \big(I_2\otimes (q_{\Hada}^{\infty})^{\top}\big)v^{x}[k],
			\]
			which is equivalent to the dynamics of \(\Psi(q_{\Hada}^{\infty})\).
			Thus, by Theorem \ref{t:synctheorem} and setting \(q=q_{\Hada}^{\infty}\), we can generate a time scale while maximizing long-term frequency stability. This solves Problem 2(c).
			
			\section{Illustrative Example}
			\label{s:example}
			\noindent In this section, we present the simulation results to validate our proposed algorithm for time scale generation.
			We consider an ensemble of \(10\) atomic clocks with \(7\) Cs-type clocks and \(3\) Hm-type clocks, i.e., \(N=10\), \(M=3\). 
			The standard deviations of the noises in Cs-type clocks and Hm-type clocks are given in Table \ref{table:Cs} and Table \ref{table:Hm}, respectively.
			The measurement noise parameter is selected as \(r=10^{-27}\).
			The sampling interval is \(\tau=1\)[s] and  \(V=\mat{I_{N-1} & -\one_{N-1}}\).
			We take \(\gamma = 0.1 \).
			
			\begin{table}[t]
				\small
				\caption{Standard dev. of noises of Cs-type clocks}
				\label{table:Cs}
				\centering
				\scalebox{0.9}{
					\begin{tabular}{lccccccccccccccc}
						& scale & 1  &  2 &   3  &  4  &  5  &  6  &  7  \\
						\hline
						$\sigma_1$ & $ 10^{-9}$ & 0.17  &  0.088 &   0.122  &  0.127  &  0.218  &  0.106  &  0.18 \\
						$\sigma_2$ &  $ 10^{-12}$ &  0.15  &  0.053  &  0.016  &  0.077  &  0.294  &  0.049  &  0.04\\  
						\hline
				\end{tabular}}
			\end{table}
			%The standard deviation of noises of Hm-type clocks is given in Table \ref{table:Hm}.
			
			\begin{table}[t]
				\centering
				\small
				\caption{Standard dev. of noises of Hm-type clocks}
				\label{table:Hm}
				\scalebox{0.9}{
					\begin{tabular}{lccccccccccccccc}
						& scale & 8  &  9 &   10  \\
						\hline
						$\sigma_1$ & $ 10^{-9}$ & 0.0216  &  0.0093  &  0.01801  \\
						$\sigma_2$ &  $ 10^{-12}$ & 0.0829  &  0.0520  &  0.0566  \\  
						$\sigma_3$ &  $ 10^{-19}$ &  1  &  1  &  1.7  \\  \hline
				\end{tabular}}
			\end{table}

			\begin{figure}[t]
				\centering
				\includegraphics[width = .99\linewidth, height= 0.6\linewidth]{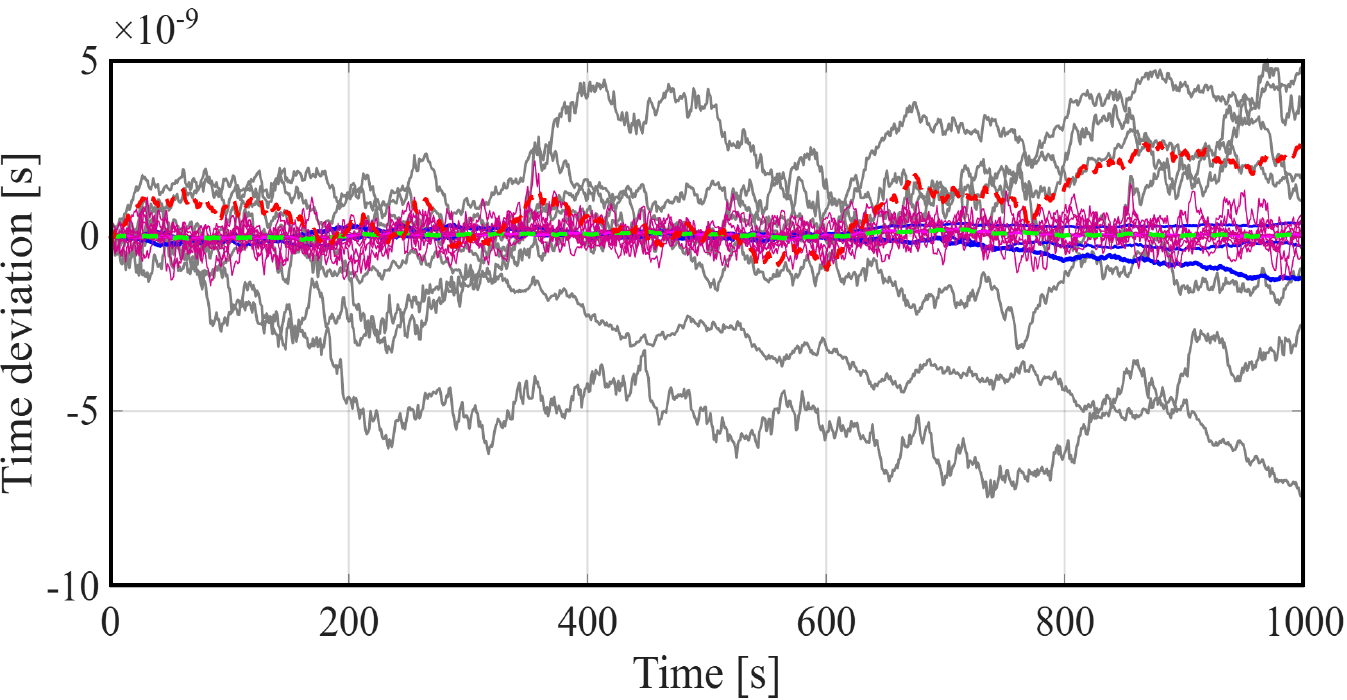}
				\caption{\small
					Time deviations sequence for short period of time when \(q= q_{\Hada}^{0}\).
					The gray lines are the time deviations of free-running Cs-type clocks.
					The blue lines are the time deviations of the free-running Hm-type clocks.
					The green line is the virtual trajectory generated by \(\Psi(q_{\Hada}^{0})\). 
					The red line is the virtual trajectory generated by \(\Psi(q_{\Hada}^{\infty})\).
					The magenta lines are the time deviations of the controlled clocks.
				}
				\label{f:improveshort-term timedeviation}
			\end{figure}
			Firstly, to generate a time scale with least short-term HVAR, we choose the weight vector \(q= q_{\Hada}^{0}\) given in \eqref{e:shorttermweight}. 
			In Fig.~\ref{f:improveshort-term timedeviation}, we show the time deviation plots for shorter period of \(1000\) [s]. 
			For reference, the time deviation plots of the free-running Cs-type clocks and free-running Hm-type clocks are also shown via gray and blue lines, respectively. 
			The virtual trajectory of \(h_g[k]\) for \(\Psi(q_{\Hada}^{0})\) and \(\Psi(q_{\Hada}^{\infty})\) are also shown in green and red lines, respectively. 
			The time deviations of the controlled clocks are shown via magenta lines.
			It is easy to see that the time deviation of each controlled clock approximately follows \(h_g[k]\) corresponding to \(\Psi(q_{\Hada}^{0})\), generating a time scale with much less divergence with respect to ideal time in shorter period of time.
			In Fig.~\ref{f:improveshort-term}, we show the HVAR curves of the controlled clocks to examine the frequency stability of the GTS. 
			For reference, the HVAR curves of the free-running Cs-type clocks and free-running Hm-type clocks are also shown.
			We also show the virtual trajectories of the HVAR curves of the best reference time scale in short-term \(\sigma_{\Hada}^2(\tau; \Psi(q_{\Hada}^{0}))\) and best reference time scale in long-term  \(\sigma_{\Hada}^2(\tau;\Psi(q_{\Hada}^{\infty}))\), respectively.
			As expected, Fig.~\ref{f:improveshort-term} illustrates that the GTS given by the controlled clocks has better short-term frequency stability as compared to individual free-running Cs-type and Hm-type clocks.
			
			\begin{figure}[t]
				\centering
				\includegraphics[width = .99\linewidth, height= 0.6\linewidth]{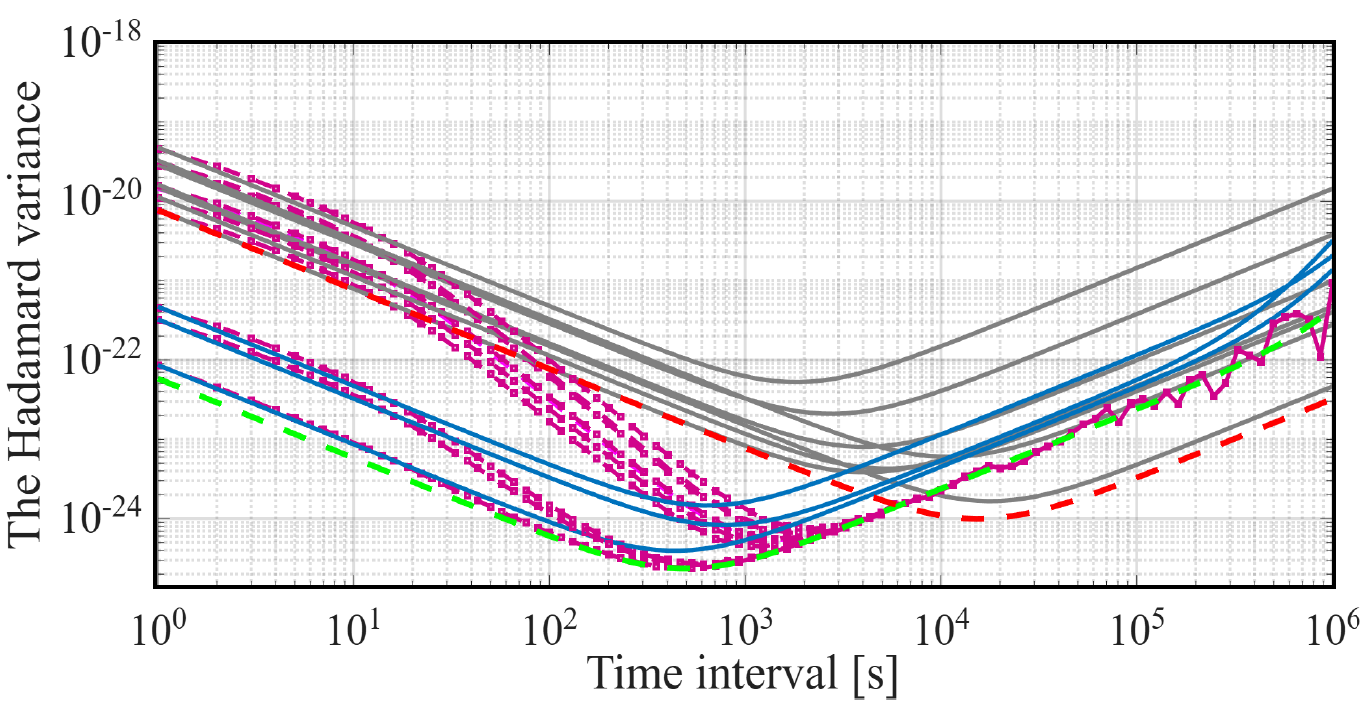}
				\caption{\small
					The Hadamard variances when \(q= q_{\Hada}^{0}\).
					The gray lines are the HVAR of the free-running Cs-type clocks.
					The blue lines are the HVAR of the free-running Hm-type clocks.
					The green line is the  HVAR of \(\Psi(q_{\Hada}^{0})\). 
					The red line is the HVAR of \(\Psi(q_{\Hada}^{\infty})\).
					The magenta lines are the HVAR of the controlled clocks.
				}
				\label{f:improveshort-term}
			\end{figure}
			
			Secondly, to generate time scale with excellent long-term frequency stability, we choose the weight vector \(q=q_{\Hada}^{\infty}\) given in \eqref{e:longtermweight}. The time deviation of each controlled clock for a duration of \(10^5\) [s] and its HVAR curve are shown in Fig.~\ref{f:timedevlong-term} and Fig.~\ref{f:improvelong-term}, respectively.
			Observe that the time deviation of each controlled clock follows the dynamics of \(h_g[k]\) in \(\Psi(q_{\Hada}^{\infty})\), generating a time scale with reduced divergence from the ideal time in longer period of time. Also, Fig.~\ref{f:improvelong-term} clearly demonstrates that the GTS has  excellent long-term frequency stability.
			
			\begin{figure}[t]
				\centering
				\includegraphics[width = .99\linewidth, height= 0.6\linewidth]{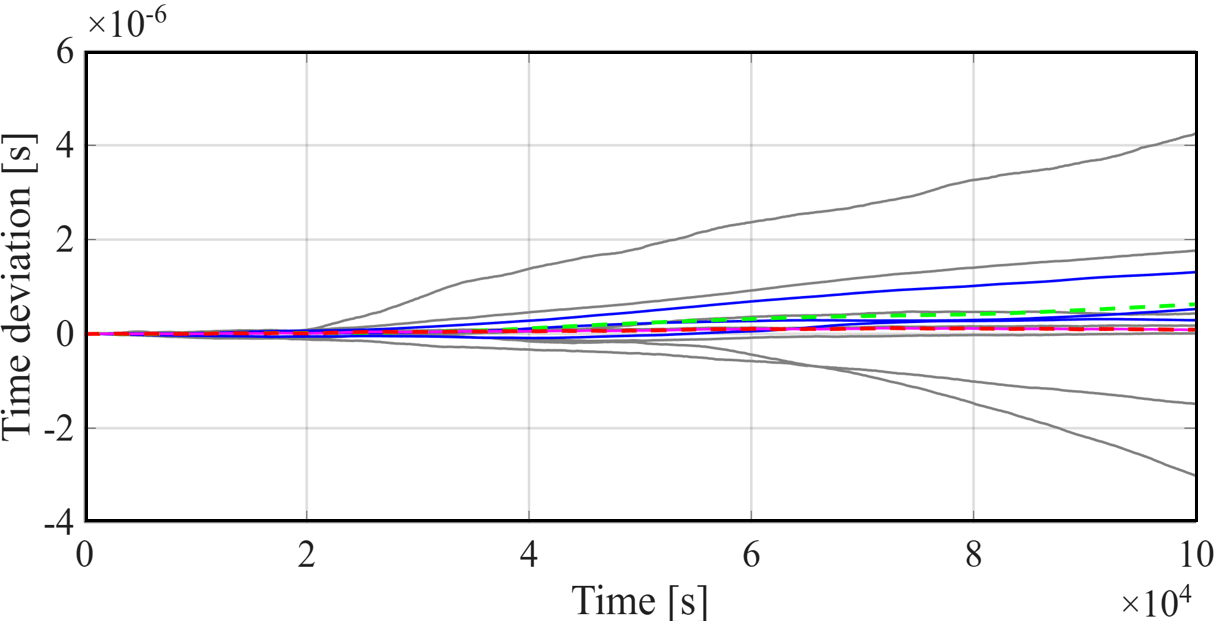}
				\caption{\small
					Time deviations sequence for long period of time when \(q= q_{\Hada}^{\infty}\).
					The gray lines are the time deviations of free-running Cs-type clocks.
					The blue lines are the time deviations of the free-running Hm-type clocks.
					The green line is the virtual trajectory generated by \(\Psi(q_{\Hada}^{0})\). 
					The red line is the virtual trajectory generated by \(\Psi(q_{\Hada}^{\infty})\).
					The magenta lines are the time deviations of the controlled clocks.
				}
				\label{f:timedevlong-term}
			\end{figure}
			
			\begin{figure}[t]
				\centering
				\includegraphics[width = .99\linewidth, height= 0.6\linewidth]{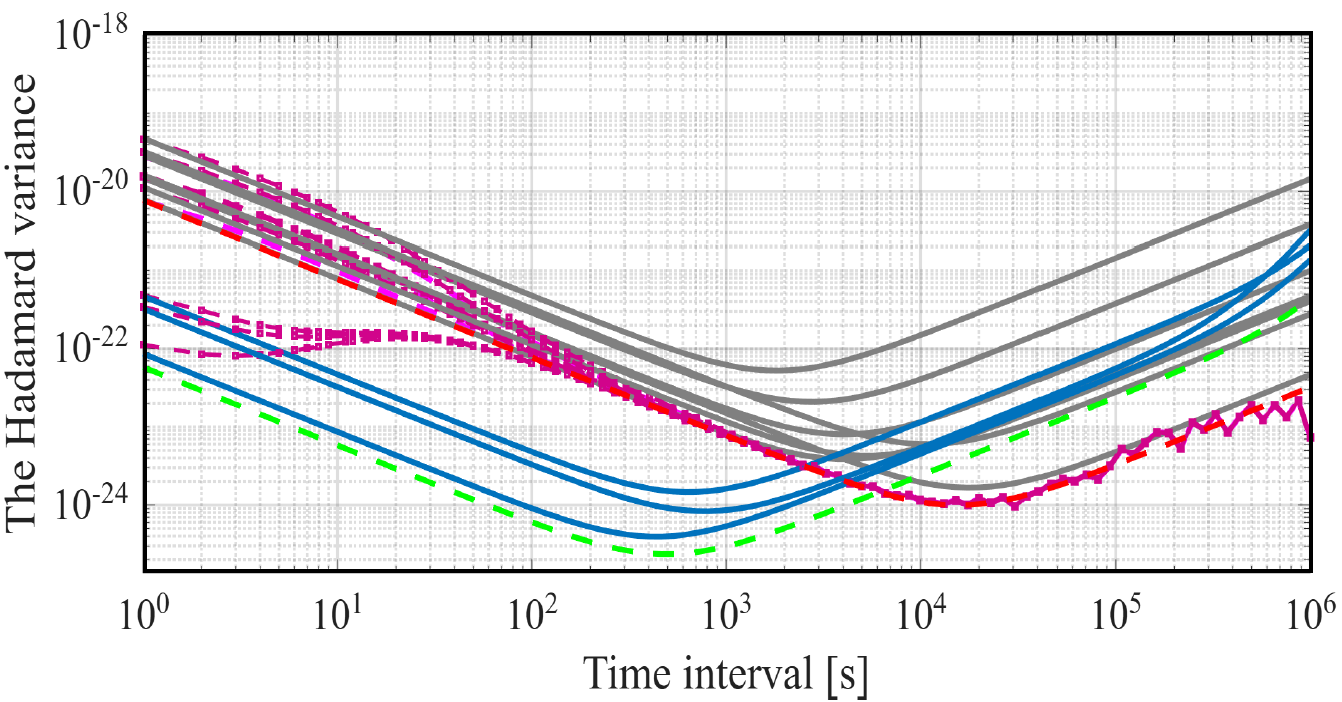}
				\caption{\small
					The Hadamard variances when \(q= q_{\Hada}^{\infty}\).
					The gray lines are the HVAR of the free-running Cs-type clocks.
					The blue lines are the HVAR of the free-running Hm-type clocks.
					The green line is the  HVAR of \(\Psi(q_{\Hada}^{0})\). 
					The red line is the HVAR of \(\Psi(q_{\Hada}^{\infty})\).
					The magenta lines are the HVAR of the controlled clocks.
				}
				\label{f:improvelong-term}
			\end{figure}
			
			Note that since \(q_{\Hada}^{0}\) and \(q_{\Hada}^{\infty}\) are not related to each other, we cannot achieve short-term and long-term frequency stability simultaneously. 
			Future work includes to improve our strategy so as to guarantee that the resulting time scale has balanced short-term and long-term frequency stability.

			%Future work involves finding a strategy to simultaneously improve the short-term and long-term stability of the time scale.

			%%%%%%%%%%%%%%%%%%%%%%%%%%%%%%%%%%%%%%%%%%%%%%%%%%%%%%%%%%%%%%%%%%%%%%%%%%%%%%%%%%%%%%%%%%%%
			
			%% There are a number of predefined theorem-like environments in
			%% ifacconf.cls:
			%%
			%% \begin{thm} ... \end{thm}            % Theorem
			%% \begin{lem} ... \end{lem}            % Lemma
			%% \begin{claim} ... \end{claim}        % Claim
			%% \begin{conj} ... \end{conj}          % Conjecture
			%% \begin{cor} ... \end{cor}            % Corollary
			%% \begin{fact} ... \end{fact}          % Fact
			%% \begin{hypo} ... \end{hypo}          % Hypothesis
			%% \begin{prop} ... \end{prop}          % Proposition
			%% \begin{crit} ... \end{crit}          % Criterion

			\bibliographystyle{IEEEtran}
			\bibliography{Ref_arxiv}

		\end{document}